\documentclass[12pt,a4paper]{article}
\usepackage[utf8]{inputenc}
\usepackage[T1]{fontenc}
\usepackage[english]{babel}
\usepackage[margin=20mm]{geometry}
\usepackage{graphicx}
\usepackage[pdftex,hidelinks]{hyperref}
\usepackage{amsmath,amssymb,amsthm}
\usepackage{mathabx}
\usepackage[numbers,sort&compress]{natbib}

\newtheorem{proposition}{Proposition}
\newtheorem{example}{Example}
\newtheorem{remark}{Remark}
\newtheorem{lemma}{Lemma}
\newtheorem{corollary}{Corollary}

\begin{document}

\title{Impact of Courant number on the results of numerical simulating of signal propagation in non-dispersive homogeneous media\footnote{This document is an English translation of the original article published in \textit{Proceedings of the Komi Science Centre of the Ural Branch of the Russian Academy of Sciences. Series ``Physical and Mathematical Sciences''}. The original publication is available at DOI: \href{https://dx.doi.org/10.19110/1994-5655-2024-5-73-83}{10.19110/1994-5655-2024-5-73-83}.}}
\author{P.A. Makarov\footnote{makarovpa@ipm.komisc.ru}, R.N. Skandakov, V.A. Ustyugov, V.I. Shcheglov}
\date{}

\maketitle

\abstract{
The paper is devoted to the study of the connection between the numerical dispersion arising in FDTD modeling of electromagnetic signal propagation in nondispersive homogeneous media optically different from vacuum and the Courant number in the 2D case.
The main results are formulated in the form of four statements, as well as a number of corollaries and remarks that determine the nature of the numerical dispersion, the optimal value of the Courant number and the limitations of the method.
It is proved that the optimal choice of the Courant number eliminates the numerical dispersion and extends the capabilities of the developed numerical algorithm to media, which refractive index lesser than refractive index of vacuum, as well as media with negative refraction.
}

\section*{Introduction}
\label{sect:Intro}

Numerical methods for solving wave equations play an important role not only in specific engineering applications, but also in fundamental science as a whole.
The FDTD (Finite-Difference Time-Domain) method~\cite{Yee:1966} belongs to such methods, and some of its features are the subject of the present work.

The main advantage of the FDTD method is the simplicity of the implementation of the computational algorithm.
This is precisely what determines the widespread use of FDTD in a wide variety of applications: biology and medicine~\cite{Miyazaki:2009,Tan:2013,Stark:2016,Nzao:2022}, ecology, geology and mineralogy~\cite{Glubokovskikh:2016,Yu:2017}, optics, photonics, electronics, communications and telecommunications~\cite{Fantoni:2017,Mishra:2019,Mohanty:2019,Bakirtzis:2021,Makarov:2022a,Makarov:2022b}.
In addition to numerous papers related in one way or another to the FDTD method, there is also extensive educational literature on this topic~\cite{Schneider:2010,Inan:2011,Taflove:2013,Langtangen:2017}.

Despite its long history of development, researchers' attention continues to be attracted by the fundamental foundations of the FDTD method.
Among these foundations is the issue of assessing the correctness of solutions obtained by the FDTD method in various problem formulations for signals with different spectral shapes~\cite{Schneider:2010,Langtangen:2017,Makarov:2023}.

It is well known (see, for example, the textbooks~\cite{Schneider:2010,Inan:2011,Taflove:2013,Langtangen:2017}) that the main parameter governing the accuracy of FDTD computations is the Courant number, which for the 2D case (1 spatial + 1 temporal dimension) has the form
\begin{equation}\label{eq:Courant-Num-Def}
  S_\mathrm{c} = \frac{c\Delta_t}{\Delta_x}
\end{equation}
and combines the main physical parameter~$c$ --- the speed of light in vacuum --- with the numerical parameters of the problem~$\Delta_x$ and~$\Delta_t$, which determine the space-time discretization step.

\begin{remark}
The quality of the numerical solution is determined not only by the choice of the Courant number~$S_\mathrm{c}$, but also by the requirements on the spectral composition of the signal, as well as by the level of the initial current of its source (see Statement~2 in our previous work~\cite{Makarov:2023}).
The influence of the latter two factors was precisely the subject of paper~\cite{Makarov:2023}, in which the Courant number was chosen optimally for modeling electromagnetic processes in vacuum, namely, it was set \mbox{$S_\mathrm{c}=1$}.
The present study is a logical continuation of~\cite{Makarov:2023}.
\end{remark}

\section{Motivation and aim of the work}
\label{sect:Motivation_and_Aim}

The choice \mbox{$S_\mathrm{c}=1$} does not, in general, ensure the quality of the resulting numerical solution in homogeneous nondispersive media optically different from vacuum.
This can be easily seen by examining Figs.~\ref{fig:1} and~\ref{fig:2}.

\begin{figure}
\center
\includegraphics[width=0.75\linewidth]{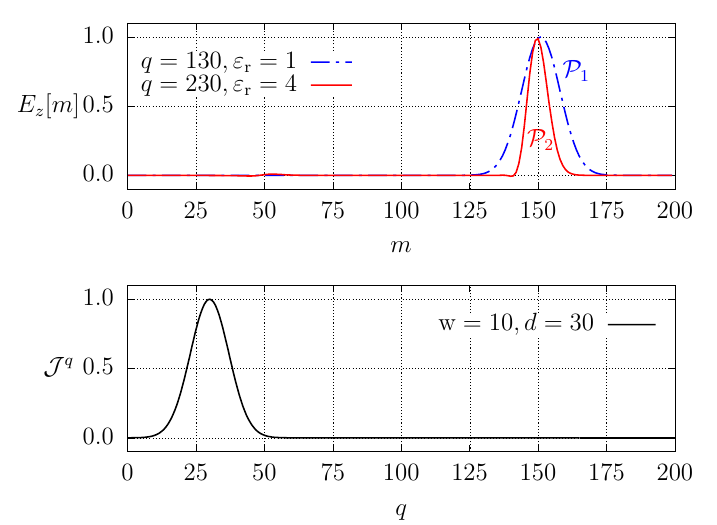}
\caption{Simulating of the propagation of Gaussian form pulses in vacuum~$\mathcal{P}_1$ and dielectric~$\mathcal{P}_2$ with relative permittivity \mbox{$\varepsilon_\mathrm{r}=4$}.}
\label{fig:1}
\end{figure}

Figs.~\ref{fig:1} and~\ref{fig:2} are constructed in accordance with the Yee algorithm, which was discussed in detail by us in the previous work~\cite{Makarov:2023} (see formulas (13), (19)--(21) therein), with the difference that now the update equations (19)--(20) explicitly take into account the material parameters of the medium (namely, its dielectric permittivity~$\varepsilon_\mathrm{r}$ and magnetic permeability~$\mu_\mathrm{r}$) in which the signals propagate:
\begin{equation}\label{eq:Maxwell-Hy}
H_y^{q+\frac{1}{2}}\left[m+\frac{1}{2}\right] =
H_y^{q-\frac{1}{2}}\left[m+\frac{1}{2}\right] +
\frac{S_\mathrm{c}}{\eta\mu_\mathrm{r}}\left(
E_z^q[m+1] - E_z^q[m]\right),
\end{equation}
\begin{equation}\label{eq:Maxwell-Ez}
E_z^{q+1}[m] = E_z^{q}[m] - \mathcal{J}^{q+\frac{1}{2}}[m] +
\frac{S_\mathrm{c}\eta}{\varepsilon_\mathrm{r}}\left(
H_y^{q+\frac{1}{2}}\left[m+\frac{1}{2}\right] -
H_y^{q+\frac{1}{2}}\left[m-\frac{1}{2}\right]\right),
\end{equation}
and the simplest absorbing boundary conditions~(21) from~\cite{Makarov:2023} are replaced by us with more general equations constructed on the basis of second-order accurate differential equations for the advection of the electromagnetic field.
All the details of the implementation of this scheme are described in detail in~\cite{Makarov:2023} and the textbook~\cite{Schneider:2010}.

\begin{remark}
Modeling the operation of a directional current source in the TF/SF formalism in this article has also undergone some changes compared to~\cite{Makarov:2023} and reduces to computing the fields according to the scheme
\begin{equation}\label{eq:Source-Hy}
H_y^{q+\frac12}\!\left[s-\frac12\right] = H_y^{q-\frac12}\!\left[s-\frac12\right] -
\frac{S_\mathrm{c}}{\eta\mu_\mathrm{r}} E_z^\mathrm{inc}\left[0,q\right],
\end{equation}
\begin{equation}\label{eq:Source-Ez}
E_z^{q+1}[s] = E_z^{q}[s] +
\frac{S_\mathrm{c}}{\sqrt{\varepsilon_\mathrm{r}\mu_\mathrm{r}}}
E_z^\mathrm{inc}\!\left[-\frac12,q+\frac12\right].
\end{equation}
Here (as in~\cite{Makarov:2023}) \mbox{$s=50$} is the fixed grid node number in all numerical experiments, specifying the spatial location of the point antenna that forms the incident field~$E_z^\mathrm{inc}$.
Computing the electric~$E_z[s]$ and magnetic~$H_y\left[s-\frac12\right]$ fields according to~\eqref{eq:Source-Ez}, \eqref{eq:Source-Hy} for a given type of incident wave~$E_z^\mathrm{inc}$ allows one to simulate the operation of the current source~$\mathcal{J}^q$ that forms the wave radiated into the right region of the grid \mbox{$m\geqslant s$}.
Hereinafter, by the operation of the current source~$\mathcal{J}^q$ we mean exactly this scheme.
\end{remark}

\begin{figure}
\center
\includegraphics[width=0.75\linewidth]{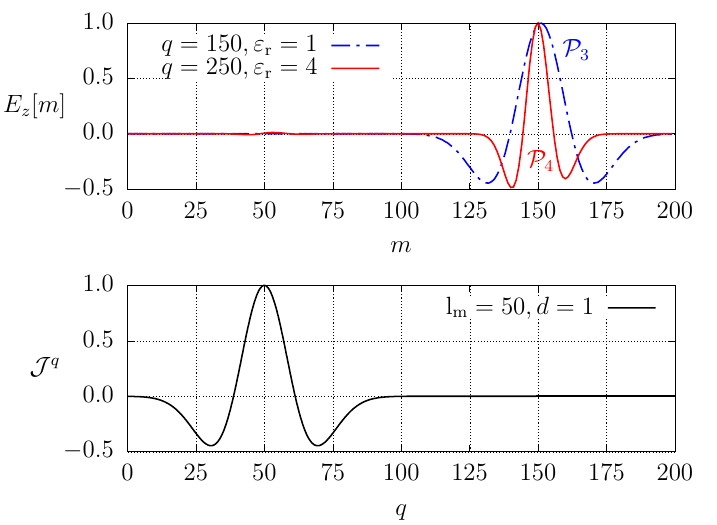}
\caption{Simulating of pulse propagation in the form of a Ricker wavelet in vacuum~$\mathcal{P}_3$ and dielectric~$\mathcal{P}_4$ with relative permittivity \mbox{$\varepsilon_\mathrm{r}=4$}.}
\label{fig:2}
\end{figure}

The parameters of the current sources~$\mathcal{J}^q$ that form the signals shown in Figs.~\ref{fig:1} and~\ref{fig:2} are also described in detail in~\cite{Makarov:2023} (see formulas (32)--(34) and the corresponding text therein), and are chosen such that the pulses~$\mathcal{P}_1$ and~$\mathcal{P}_3$ propagating in vacuum can be considered a correct numerical solution of the problem in the sense of Definition~2 given in~\cite{Makarov:2023}.
Proceeding from this, the pulses~$\mathcal{P}_1$ and~$\mathcal{P}_3$ can be considered ``reference'' ones, comparing with which the pulses~$\mathcal{P}_2$ and~$\mathcal{P}_4$, respectively, one can easily assess the influence of the medium on the correctness of the numerical solution.

Figs.~\ref{fig:1} and~\ref{fig:2} are constructed for the same value of the Courant number \mbox{$S_\mathrm{c}=1$}, for which the source~\cite{Schneider:2010} states that it minimizes FDTD numerical errors, and the textbook~\cite{Langtangen:2017} claims that such a choice allows one to obtain an exact solution of the problem.

At the same time, it is obvious that the FDTD solutions depicted by the pulses~$\mathcal{P}_2$ in Fig.~\ref{fig:1} and~$\mathcal{P}_4$ in Fig.~\ref{fig:2} are not correct in the sense of Definition~2 given by us in~\cite{Makarov:2023}, which is in direct contradiction with what was noted in the previous paragraph.
This incorrectness is the result of a phenomenon that occurs everywhere in FDTD numerical calculations and is known in the literature~\cite{Schneider:2010,Langtangen:2017} as ``numerical dispersion''.

Indeed, the signals~$\mathcal{P}_2$ in Fig.~\ref{fig:1} and~$\mathcal{P}_4$ in Fig.~\ref{fig:2} do not represent either the original Gaussian pulse or the Ricker wavelet, respectively.
During the propagation of these wave packets in a dielectric with \mbox{$\varepsilon_\mathrm{r}=4$} over a sufficiently long time (\mbox{$\Delta q=230$} and~$250$, respectively), their shape is significantly distorted (note that this effect is rather weak for the parameters used and begins to manifest itself clearly only towards the end of the simulation).
This contradicts the initial mathematical model of the phenomenon we are simulating, since the medium was assumed to be nondispersive.
This obvious contradiction raises a natural question --- can the FDTD method be used at all to obtain correct simulation results for signal propagation in nondispersive materials?
And if this is possible, then how does the choice of the Courant number~$S_\mathrm{c}$ affect the correctness of the obtained solutions?

Despite the contradiction noted above, the FDTD method was previously used by us in a number of works that considered either wave propagation in randomly inhomogeneous media~\cite{Makarov:2022a,Makarov:2022b}, or the features of the solution of the homogeneous and inhomogeneous Cauchy problems by the FDTD method~\cite{Makarov:2023}.
In both the first and the second case, all problems associated with numerical dispersion were completely ignored, partly because in all these works a significant part of the path along which the propagation of electromagnetic signals was considered was air space (modeled by us as indistinguishable from vacuum).
At the same time, such an ignoring of this aspect of the matter cannot be fully justified, which requires, if not a complete correction of numerical dispersion, then at least a clarification of the effects associated with it.

Further, note that in the book~\cite{Schneider:2010} it is postulated without any proof that the value \mbox{$S_\mathrm{c}=1$} is the maximum possible, which, generally speaking, does not correspond at all to the meaning of definition~\eqref{eq:Courant-Num-Def}, which does not impose any restrictions on the arbitrarily chosen space-time grid steps~$\Delta_x$ and~$\Delta_t$, and the ratio between them.
The source~\cite{Langtangen:2017} more cautiously states that the choice \mbox{$S_\mathrm{c}>1$} leads to an exponential growth of noise caused by rounding errors, which are always present in numerical simulation, which ultimately completely destroys the obtained solution.

\begin{remark}
It should also be noted that in the works~\cite{Schneider:2010} and~\cite{Langtangen:2017} different quantities are called the Courant numbers, which creates additional difficulties in understanding all the noted issues.
This is because in~\cite{Schneider:2010} in definition~\eqref{eq:Courant-Num-Def} the quantity~$c$ is the speed of light in vacuum (as is also adopted by us in this work), while in~\cite{Langtangen:2017} it is considered that~$c$ is the speed of wave propagation in a given specific medium.
\end{remark}

In addition to the problems already noted, the question also arises about the possibility of applying the FDTD method for modeling electrodynamics in ``not the most ordinary media'' (even if the phenomenon of dispersion is neglected).
The following examples are meant here: the propagation of radio waves in the Earth's ionosphere~\cite{Vinogradova}, electromagnetic waves of the terahertz or X-ray range in conductors, semiconductors or dielectrics~\cite{Kugushev,Bredov}, media with negative values of relative permittivities (in particular, ``left-handed'' media~\cite{Shuster,Veselago,Pendry,Agranovich}).
In all these cases, the dispersion of electromagnetic waves plays a decisive role, the anomalous nature of which from the mathematical point of view consists in the fact that the permittivity~$\varepsilon_\mathrm{r}$ and permeability~$\mu_\mathrm{r}$ can take values less than unity, and even more so --- negative ones.

At the same time, the usual algorithms found in the literature on the FDTD method~\cite{Yee:1966,Schneider:2010,Inan:2011,Taflove:2013,Langtangen:2017,Makarov:2022a,Makarov:2022b,Makarov:2023} do not allow one to simulate the propagation of electromagnetic waves in such media simply by setting the values \mbox{$0<\varepsilon_\mathrm{r},\mu_\mathrm{r}<1$} and \mbox{$\varepsilon_\mathrm{r},\mu_\mathrm{r}<0$}.

The present article is devoted to solving the noted issues, the aim of which is thus to consider the numerical dispersion of the FDTD method and to assess the influence of the choice of the Courant number values on the quality of modeling the propagation of signals in homogeneous nondispersive media.

\section{Numerical dispersion in the FDTD method}

To obtain an expression describing the numerical dispersion in the Yee grid, let us return to the original discrete analogs of the Ampere~\eqref{eq:Maxwell-Ez} and Faraday~\eqref{eq:Maxwell-Hy} equations, which in a more compact and convenient form can be written using the shift operators in the space-time grid~$\widehat{\mathcal{S}}_x^\chi$ and~$\widehat{\mathcal{S}}_t^\tau$ (here~$\chi$ and~$\tau$ are shift parameters having the form \mbox{$\chi,\tau=p/2,\,\forall p\in\mathbb{Z}$}).
By definition, the action of these operators has the form
\begin{equation}\label{eq:Sx-Def}
\widehat{\mathcal{S}}_x^\chi: \, \widehat{\mathcal{S}}_x^\chi \psi^q[m] = \psi^q\left[m+\chi\right],
\end{equation}
\begin{equation}\label{eq:St-Def}
\widehat{\mathcal{S}}_t^\tau: \, \widehat{\mathcal{S}}_t^\tau \psi^q[m] = \psi^{q+\tau}[m],
\end{equation}
where~$\psi$ denotes an arbitrary component of the electromagnetic field (in the 2D case considered by us~--- $E_z$ or~$H_y$).

It is easy to verify that with the help of the operators~\eqref{eq:Sx-Def} and~\eqref{eq:St-Def} the Ampere equation~\eqref{eq:Maxwell-Ez} in the absence of extraneous sources (\mbox{$\mathcal{J}=0$}) can be written in the form
\begin{equation}\label{eq:Ampere-Law-Discrete-1}
\begin{aligned}
\widehat{\mathcal{S}}_t^{\frac12} \varepsilon
\left( \frac{\widehat{\mathcal{S}}_t^{\frac12} -
\widehat{\mathcal{S}}_t^{-\frac12}}{\Delta_t} \right) E_z^q[m] =
\widehat{\mathcal{S}}_t^{\frac12}
\left( \frac{\widehat{\mathcal{S}}_x^{\frac12} -
\widehat{\mathcal{S}}_x^{-\frac12}}{\Delta_x} \right) H_y^q[m].
\end{aligned}
\end{equation}
Here we take into account the definition of the Courant number~\eqref{eq:Courant-Num-Def}, the relation \mbox{$c=1/\sqrt{\varepsilon_0\mu_0}$} of the speed of light in vacuum with its dielectric permittivity~$\varepsilon_0$ and magnetic permeability~$\mu_0$ and the expression for the characteristic impedance of vacuum \mbox{$\eta=\sqrt{\mu_0/\varepsilon_0}\approx120\pi$} (for reference see, for example, the textbooks~\cite{Bredov,Kugushev}).
In addition, when writing~\eqref{eq:Ampere-Law-Discrete-1}, the notation for the absolute dielectric permittivity of the medium \mbox{$\varepsilon=\varepsilon_\mathrm{r}\varepsilon_0$} is used.

For convenience, let us also define the finite-difference operators~$\tilde{\partial}_i$ according to
\begin{equation}\label{eq:Fin-Diff-Operator-Def}
\tilde{\partial}_i =
\frac{\widehat{\mathcal{S}}_i^{\frac12} - \widehat{\mathcal{S}}_i^{-\frac12}}{\Delta_i},
\end{equation}
where~$i$ is either~$x$ or~$t$.
With the help of~\eqref{eq:Fin-Diff-Operator-Def} the Ampere law in a stationary medium without extraneous currents~\eqref{eq:Ampere-Law-Discrete-1} can finally be written in the Yee form~\cite{Schneider:2010}
\begin{equation}\label{eq:Ampere-Law-Yee}
\varepsilon \widehat{\mathcal{S}}_t^{\frac12} \tilde{\partial}_t E_z^q[m] =
\widehat{\mathcal{S}}_t^{\frac12} \tilde{\partial}_x H_y^q[m].
\end{equation}
In the same way, using~\eqref{eq:Sx-Def}, \eqref{eq:St-Def} and~\eqref{eq:Fin-Diff-Operator-Def} the discrete analog of the Faraday law~\eqref{eq:Maxwell-Hy} is reduced to the Yee form
\begin{equation}\label{eq:Faraday-Law-Yee}
\mu \widehat{\mathcal{S}}_x^{\frac12} \tilde{\partial}_t H_y^q[m] =
\widehat{\mathcal{S}}_x^{\frac12} \tilde{\partial}_x E_z^q[m].
\end{equation}

Now consider a plane monochromatic wave of frequency~$\omega$ propagating in the Yee grid to the right (\mbox{$m\geqslant s$})
\begin{equation}\label{eq:Plane-Wave-Ez}
E_z^q[m] = E_0 \, e^{i(\omega q\Delta_t - \tilde{\beta}m\Delta_x)},
\end{equation}
\begin{equation}\label{eq:Plane-Wave-Hy}
H_y^q[m] = H_0 \, e^{i(\omega q\Delta_t - \tilde{\beta}m\Delta_x)},
\end{equation}
where~$\tilde{\beta}$ is the wavenumber, i.e., the propagation constant of a plane monochromatic wave in the FDTD grid (different from the corresponding constant~$\beta$ in continuous space), and~$E_0$ and~$H_0$ are the complex amplitudes of the electric and magnetic field intensities.

\begin{lemma}\label{Lemma:1}
The action of the operators~\eqref{eq:Fin-Diff-Operator-Def} on an arbitrary component~$\psi$ of a plane monochromatic wave~\eqref{eq:Plane-Wave-Ez}, \eqref{eq:Plane-Wave-Hy} reduces to
\begin{equation}\label{eq:Fin-Diff-Operator-Dt}
\tilde{\partial}_t \psi =
i \frac{2}{\Delta_t} \sin \left( \frac{\omega\Delta_t}{2} \right) \psi,
\end{equation}
\begin{equation}\label{eq:Fin-Diff-Operator-Dx}
\tilde{\partial}_x \psi =
-i \frac{2}{\Delta_x} \sin \left( \frac{\tilde{\beta}\Delta_x}{2} \right) \psi.
\end{equation}
\end{lemma}

\begin{proof}
It is verified by direct substitution of~\eqref{eq:Plane-Wave-Ez}, \eqref{eq:Plane-Wave-Hy} into~\eqref{eq:Fin-Diff-Operator-Def} taking into account~\eqref{eq:Sx-Def}, \eqref{eq:St-Def}.
Let us write out here in explicit form only the action on~$\psi$ of the shift operators~\eqref{eq:Sx-Def}, \eqref{eq:St-Def} with parameters \mbox{$\chi,\tau=\pm\frac12$}
\begin{equation}\label{eq:Sx-St-Plane-Wave}
\widehat{\mathcal{S}}_t^{\pm\frac12} \psi =
e^{\pm i\omega\Delta_t/2} \psi,\quad
\widehat{\mathcal{S}}_x^{\pm\frac12} \psi =
e^{\mp i\tilde{\beta}\Delta_x/2} \psi,
\end{equation}
with the help of which the equalities~\eqref{eq:Fin-Diff-Operator-Dt} and~\eqref{eq:Fin-Diff-Operator-Dx} are obtained elementarily.
\end{proof}

\begin{proposition}
The dispersion relation for the Yee grid can be represented in the form
\begin{equation}\label{eq:Yee-Dispersion-Eq}
\sin \left( \frac{\omega\Delta_t}{2} \right) =
\frac{\Delta_t}{\sqrt{\varepsilon\mu}\Delta_x}
\sin \left( \frac{\tilde{\beta}\Delta_x}{2} \right).
\end{equation}
\end{proposition}

\begin{proof}
Using the results~\eqref{eq:Fin-Diff-Operator-Dt} and~\eqref{eq:Fin-Diff-Operator-Dx} of Lemma~\ref{Lemma:1}, as well as~\eqref{eq:Sx-St-Plane-Wave} in the Ampere law~\eqref{eq:Ampere-Law-Yee} for a plane monochromatic wave, we write
\begin{equation}
\begin{aligned}
i \varepsilon \frac{2}{\Delta_t} \sin \left( \frac{\omega\Delta_t}{2} \right)
e^{i\omega\Delta_t/2} E_z^q[m] =
-i \frac{2}{\Delta_x} \sin \left( \frac{\tilde{\beta}\Delta_x}{2} \right)
e^{i\omega\Delta_t/2} H_y^q[m].
\end{aligned}
\end{equation}
Substituting into the last equality the explicit expressions for the fields~\eqref{eq:Plane-Wave-Ez}, \eqref{eq:Plane-Wave-Hy} for a plane monochromatic wave and canceling common factors, we obtain
\begin{equation}
\varepsilon \frac{1}{\Delta_t} \sin \left( \frac{\omega\Delta_t}{2} \right) E_0 =
-\frac{1}{\Delta_x} \sin \left( \frac{\tilde{\beta}\Delta_x}{2} \right) H_0.
\end{equation}
Hence we arrive at the expression for the numerical impedance in the Yee grid
\begin{equation}\label{eq:FDTD-Impedance-1}
\frac{E_0}{H_0} = -\frac{\Delta_t}{\varepsilon\Delta_x} \cdot
\frac{\sin \cfrac{\tilde{\beta}\Delta_x}{2}}
{\sin \cfrac{\omega\Delta_t}{2}}.
\end{equation}
Performing similar actions with respect to the Faraday law~\eqref{eq:Faraday-Law-Yee}, we sequentially obtain
\begin{equation}
\begin{aligned}
&i \mu \frac{2}{\Delta_t} \sin \left( \frac{\omega\Delta_t}{2} \right)
e^{-i\tilde{\beta}\Delta_x/2} H_y^q[m] =
-i \frac{2}{\Delta_x} \sin \left( \frac{\tilde{\beta}\Delta_x}{2} \right)
e^{-i\tilde{\beta}\Delta_x/2} E_z^q[m],
\end{aligned}
\end{equation}
\begin{equation}
\mu \frac{1}{\Delta_t} \sin \left( \frac{\omega\Delta_t}{2} \right) H_0 =
-\frac{1}{\Delta_x} \sin \left( \frac{\tilde{\beta}\Delta_x}{2} \right) E_0
\end{equation}
and the corresponding impedance in the form
\begin{equation}\label{eq:FDTD-Impedance-2}
\frac{E_0}{H_0} = -\frac{\mu\Delta_x}{\Delta_t} \cdot
\frac{\sin \cfrac{\omega\Delta_t}{2}}
{\sin \cfrac{\tilde{\beta}\Delta_x}{2}}.
\end{equation}
Equating the right-hand sides of~\eqref{eq:FDTD-Impedance-1} and~\eqref{eq:FDTD-Impedance-2}, and performing cross-multiplication of the factors in the resulting equality, we obtain
\begin{equation}
\sin^2 \left( \frac{\omega\Delta_t}{2} \right) =
\frac{\Delta_t^2}{\varepsilon\mu\Delta_x^2}
\sin^2 \left( \frac{\tilde{\beta}\Delta_x}{2} \right).
\end{equation}
Taking the square root in the last equality, we finally arrive at~\eqref{eq:Yee-Dispersion-Eq}, which completes the proof.
\end{proof}

\begin{remark}
The dispersion relation~\eqref{eq:Yee-Dispersion-Eq} of the FDTD method differs significantly from its continuous analog, which has the form~\cite{Bredov,Kugushev}
\begin{equation}\label{eq:Dispersion-Classic}
\beta = \omega \sqrt{\varepsilon \mu}.
\end{equation}
At the same time, we note that this difference becomes vanishingly small with sufficiently fine space-time discretization.
Indeed, retaining in the Taylor series expansions of the sine at~$\Delta_t$ and \mbox{$\Delta_x\rightarrow0$} only the first-order terms, one can easily obtain from~\eqref{eq:Yee-Dispersion-Eq}
\begin{equation}\label{eq:Yee-Dispersion-DtDx0}
\tilde{\beta} = \omega \sqrt{\varepsilon \mu}.
\end{equation}
\end{remark}

Let us emphasize, however, that~\eqref{eq:Yee-Dispersion-DtDx0} is valid only in the limit of infinitely fine discretization \mbox{$\Delta_t,\Delta_x\rightarrow0$}.
In the general case of finite space-time discretization, the phase velocities of the wave in the Yee FDTD grid~$\tilde{c}_\mathrm{p}$ and in continuous space~$c_\mathrm{p}$ will be different.

\begin{proposition}\label{Stmnt:Phase-Velocity}
The deviation of the phase velocity of the wave in the Yee grid from the corresponding value in the continuous case can be described by the equality
\begin{equation}
\label{eq:FDTD-Dispersion}
\frac{\tilde{c}_\mathrm{p}}{c_\mathrm{p}} =
\frac{\pi\sqrt{\varepsilon_\mathrm{r}\mu_\mathrm{r}}}
{N_\lambda\arcsin\left[\cfrac{\sqrt{\varepsilon_\mathrm{r}\mu_\mathrm{r}}}{S_\mathrm{c}} \,
\sin\!\left(\cfrac{\pi S_\mathrm{c}}{N_\lambda}\right)\right]},
\end{equation}
where the parameter~$N_\lambda$ is the number of spatial grid nodes per wavelength in free space
\begin{equation}\label{eq:Lambda-Discretization}
\lambda = N_\lambda\Delta_x.
\end{equation}
\end{proposition}

\begin{proof}
We use the relation of the phase velocity with the wavenumber in continuous space~\cite{Bredov,Kugushev} and in the FDTD grid
\begin{equation}
c_\mathrm{p} = \frac{\omega}{\beta},\quad
\tilde{c}_\mathrm{p} = \frac{\omega}{\tilde{\beta}},
\end{equation}
which allows us to reduce the left-hand side of equality~\eqref{eq:FDTD-Dispersion} to the form
\begin{equation}\label{eq:Velocities-Relation-1}
\frac{\tilde{c}_\mathrm{p}}{c_\mathrm{p}} =
\frac{\beta}{\tilde{\beta}} = \frac{\frac{\beta\Delta_x}2}{\frac{\tilde{\beta}\Delta_x}2}.
\end{equation}
Further, we apply~\eqref{eq:Dispersion-Classic} to the numerator of the last equality
\begin{equation}\label{eq:Beta-1}
\beta = \omega \sqrt{\varepsilon\mu} = 2 \pi \frac{c}{\lambda}
\sqrt{\varepsilon_0\varepsilon_\mathrm{r}\mu_0\mu_\mathrm{r}} =
\frac{2 \pi}{\lambda} \sqrt{\varepsilon_\mathrm{r}\mu_\mathrm{r}},
\end{equation}
where the wavelength in vacuum~$\lambda$ is discretized according to~\eqref{eq:Lambda-Discretization}, which gives
\begin{equation}
\frac{\beta\Delta_x}2 = \frac{\pi\sqrt{\varepsilon_\mathrm{r}\mu_\mathrm{r}}}{N_\lambda}.
\end{equation}
Applying now to the denominator of the right-hand side of~\eqref{eq:Velocities-Relation-1} the dispersion relation~\eqref{eq:Yee-Dispersion-Eq}, in which transformations similar to~\eqref{eq:Beta-1} are used, as well as the definition of the Courant number~\eqref{eq:Courant-Num-Def} together with~\eqref{eq:Lambda-Discretization}, we obtain
\begin{equation}
\frac{\tilde{\beta}\Delta_x}2 = \arcsin
\left[ \frac{\sqrt{\varepsilon_\mathrm{r}\mu_\mathrm{r}}}{S_\mathrm{c}}
\sin \left(\frac{\pi S_\mathrm{c}}{N_\lambda}\right) \right].
\end{equation}
Substitution of the last two equalities into~\eqref{eq:Velocities-Relation-1} leads to the result~\eqref{eq:FDTD-Dispersion}, which completes the proof.
\end{proof}

Relation~\eqref{eq:FDTD-Dispersion} within the framework of the restrictions considered in this article (according to which nondispersive homogeneous media are investigated, for which both~$\varepsilon_\mathrm{r}$ and~$\mu_\mathrm{r}$ are some real constants), has an obvious meaning in the following domain of definition of the parameters:
\begin{equation}\label{eq:Definition-Area}
N_\lambda \in \mathbb{N} \backslash 1,\quad
\varepsilon_\mathrm{r}, \mu_\mathrm{r}, S_\mathrm{c} \in
(0,+\infty) \subset \mathbb{R},
\end{equation}
where neither the set of natural numbers greater than unity~$\mathbb{N}\backslash1$, nor the set of positive real numbers are considered to contain the element~$+\infty$ itself.
Going beyond the domain of definition~\eqref{eq:Definition-Area} requires a separate detailed study, and we will return to it in the last section of this article.

Before proceeding to the discussion of the features of expression~\eqref{eq:FDTD-Dispersion}, we note that the dielectric and magnetic permittivities of the medium enter it only in the combination
\begin{equation}
n_\mathrm{r} = \sqrt{\varepsilon_\mathrm{r}\mu_\mathrm{r}},
\end{equation}
which has the physical meaning of the relative refractive index of the medium.

\begin{example}
Consider the propagation of a plane monochromatic wave with \mbox{$N_\lambda=10$} in glass with a refractive index \mbox{$n_\mathrm{r}=1.5$} in the case of the Courant number \mbox{$S_\mathrm{c}=1$}.
The ratio~\eqref{eq:FDTD-Dispersion} in this case is \mbox{$\tilde{c}_\mathrm{p}/c_\mathrm{p}\approx0.9777$}, which in percentage terms amounts to a numerical error of about $2.23\%$.
Thus, in this situation, for each unit of path equal to the wavelength, the FDTD calculation accumulates a significant phase error of the order of $8.03^\circ$.
Note also that improving the wavelength discretization by a factor of two (\mbox{$N_\lambda=20$}) reduces the corresponding errors to the values~$0.48\%$ and~$1.89^\circ$ (i.e., approximately fourfold), as should happen for a second-order accurate computational method.
\end{example}

\begin{figure}
\center
\includegraphics[width=0.75\linewidth]{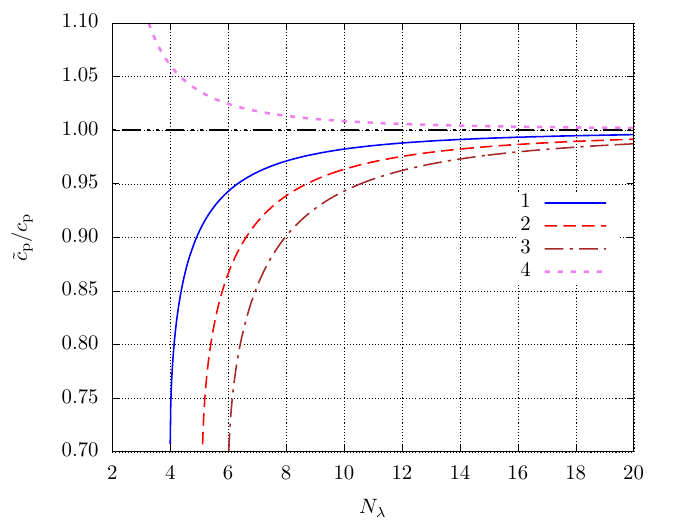}
\caption{Dependence of the phase velocity ratio, determined according to the FDTD method with respect to its exact value~(\ref{eq:FDTD-Dispersion}), on the wavelength discretization parameter~$N_\lambda$ for some media with relative refractive indices~$n_\mathrm{r}$. The Courant number \mbox{$S_\mathrm{c}=1$}.}
\label{fig:3}
\end{figure}

The features of the numerical dispersion of the FDTD method, determined by expression~\eqref{eq:FDTD-Dispersion}, and complementing the example given above, are presented in Fig.~\ref{fig:3}.
This figure shows a family of curves for which \mbox{$n_\mathrm{r}>S_\mathrm{c}$} (namely, \mbox{$n_\mathrm{r}=\sqrt{2}$} for curve~$1$, \mbox{$n_\mathrm{r}=\sqrt{3}$}~--- for line~$2$, and \mbox{$n_\mathrm{r}=2$}~--- for~$3$, respectively).
It can be seen that an increase in the refractive index of the medium with poor wavelength discretization leads to a significant lag of the wave simulated by the FDTD method, which is expressed in a significant deviation of the ratio~$\tilde{c}_\mathrm{p}/c_\mathrm{p}$ from unity.

In addition, Fig.~\ref{fig:3} shows curve~$4$, corresponding to the case \mbox{$n_\mathrm{r}<S_\mathrm{c}$} (in this particular case, the value \mbox{$n_\mathrm{r}=\sqrt{1/2}$} is chosen).
This example demonstrates that in media optically less dense than vacuum, the FDTD calculation leads to the propagation of the simulated waves with a lead compared to the true velocity.

\begin{corollary}
It can be seen from Fig.~\ref{fig:3} that the accuracy of the FDTD calculation decreases with decreasing~$N_\lambda$, and also~--- as the permittivity of the medium deviates from unity (which is the case for both dielectrics, and magnets, and magnetic dielectrics).
\end{corollary}

\begin{corollary}
It is easy to calculate the limit of the ratio~\eqref{eq:FDTD-Dispersion}, which, regardless of the values of~$n_\mathrm{r}$ and~$S_\mathrm{c}$, is equal to
\begin{equation}
\lim_{N_\lambda\rightarrow+\infty} \frac{\tilde{c}_\mathrm{p}}{c_\mathrm{p}} = 1.
\end{equation}
This means that the accuracy of the FDTD calculation improves with increasing~$N_\lambda$.
\end{corollary}

\begin{remark}
The decrease in the accuracy of the FDTD calculation in a medium with~$n_\mathrm{r}>1$ is due to the fact that the wavelength in such a medium is shorter than in vacuum, and a small discretization of the wavelength~$N_\lambda$ is not enough for a correct calculation.
However, a simple transfer of this statement to the case of a medium with~$\varepsilon_\mathrm{r}<1$ is not so obvious and requires clarification.
\end{remark}

\begin{corollary}
In addition, Fig.~\ref{fig:3} indicates that high-frequency components of a wave packet in media with \mbox{$\varepsilon_\mathrm{r}>1$} tend to lag behind, while in media with \mbox{$\varepsilon_\mathrm{r}<1$} this feature changes to the opposite~--- high-frequency components of the wave packet in the Yee grid propagate with a higher phase velocity than is actually the case.
In other words, in media optically denser compared to vacuum, the FDTD calculation leads to the occurrence of numerical dispersion having the character of anomalous dispersion~\cite{Landsberg}.
And conversely~--- when modeling optically less dense media, the influence of FDTD discretization corresponds to the behavior of normal dispersion (in this case, the ``red'' components of the wave packet ``overtake the blue'' ones).
\end{corollary}

\begin{remark}
The structure of the denominator~\eqref{eq:FDTD-Dispersion} has an obvious resemblance to the form of dispersion equations obtained in problems of wave propagation in media with periodic inhomogeneities~\cite{Vinogradova,Kotkin,Karlov,Flugge}.
Namely, the function appears here
\begin{equation}\label{eq:Phi}
\varphi(N_\lambda, S_c, n_\mathrm{r}) =
\frac{n_\mathrm{r}}{S_\mathrm{c}} \sin \frac{\pi S_\mathrm{c}}{N_\lambda},
\end{equation}
the character of which determines the passbands and stopbands of the Yee grid.
Indeed, it is easy to understand that the region of parameter values $(N_\lambda, S_\mathrm{c}, n_\mathrm{r})$ in which the condition \mbox{$\left|\varphi(N_\lambda, S_c, n_\mathrm{r})\right| > 1$} is satisfied corresponds to the non-transmission of waves, since the expression $\arcsin\varphi$ in this case has no meaning in real values.
\end{remark}

A similar picture of the phenomenon takes place in the Kronig--Penney model, as well as in the Hill and Mathieu equations describing parametric oscillations~\cite{Vinogradova,Kotkin,Karlov,Flugge}, as well as wave propagation in systems with spatial periodicity in the arrangement of inhomogeneities.
As a periodic structure in our case (in the FDTD method), the Yee computational grid itself acts directly.
In this case, the transmission and non-transmission of waves (the quantum-mechanical analog~--- allowed and forbidden bands) is determined precisely by the quantities~$N_\lambda$, $S_\mathrm{c}$ and~$n_\mathrm{r}$.

\begin{figure}
\center
\includegraphics[width=0.75\linewidth]{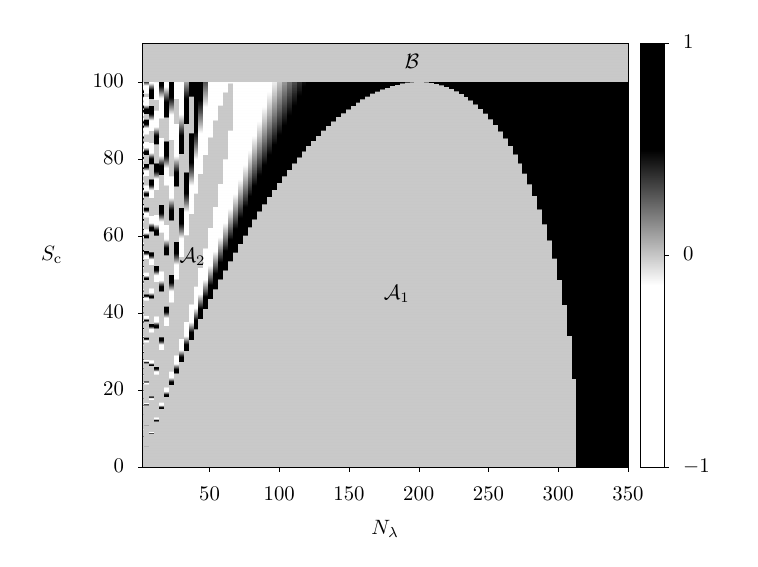}
\caption{$(N_\lambda,S_\mathrm{c})$-diagram of Yee grid bandwidths when \mbox{$n_\mathrm{r}=100$}.}
\label{fig:4}
\end{figure}

Fig.~\ref{fig:4} illustrates the $(N_\lambda,S_\mathrm{c})$-diagram of the Yee grid passbands, constructed on the basis of the function~\eqref{eq:Phi} for a relatively large value of the refractive index of the medium \mbox{$n_\mathrm{r}=100$}.
The regions of the stopbands~$\left\{\mathcal{A}_i\right\}_{i=1}^{\mathcal{N}}$, corresponding to the condition \mbox{$\left|\varphi(N_\lambda, S_c, n_\mathrm{r})\right| > 1$}, are shown in this diagram by a monotone gray shading.
The total number of such regions~$\mathcal{N}$ (the first two of them are marked in Fig.~\ref{fig:4}) depends on the value of the refractive index of the medium~$n_\mathrm{r}$, and grows with its increase.
The passbands are shown in this diagram by a gradient fill, changing from black (\mbox{$\varphi=+1$}) to white (\mbox{$\varphi=-1$}).
It is within the framework of the latter regions that the dispersion relation in the form~\eqref{eq:FDTD-Dispersion} makes sense.

Also in Fig.~\ref{fig:4} the region~$\mathcal{B}$ defined by the inequality \mbox{$S_\mathrm{c}>n_\mathrm{c}$} is highlighted in gray, within which the FDTD method also cannot provide a correct numerical solution of the problem of wave propagation in a homogeneous nondispersive medium.
The arguments confirming the validity of this statement are given in the next section of the article when discussing Fig.~\ref{fig:6}.

\section{Relation of numerical dispersion with the Courant number}

Now, after the preliminary study, let us discuss the algorithm for correcting the numerical dispersion by a special choice of the Courant number.
Its possibility is based on the relation~\eqref{eq:FDTD-Dispersion}, from which the following result directly follows.

\begin{proposition}\label{Stmnt:Magic-Courant}
For any given~$\varepsilon_\mathrm{r}$ and~$\mu_\mathrm{r}$ belonging to the domain of definition~\eqref{eq:Definition-Area}, it is always possible to eliminate the computational errors associated with numerical dispersion by setting the Courant number equal to
\begin{equation}\label{eq:Magic-Courant}
S_\mathrm{c} = \sqrt{\varepsilon_\mathrm{r} \mu_\mathrm{r}}.
\end{equation}
\end{proposition}

\begin{proof}
It is verified by direct substitution of~\eqref{eq:Magic-Courant} into~\eqref{eq:FDTD-Dispersion}.
\end{proof}

\begin{remark}
The value of the Courant number~\eqref{eq:Magic-Courant} can be called ``magic'', since in this case the phase velocity of the wave in the FDTD grid~$\tilde{c}_\mathrm{p}$ exactly coincides with the real value of the phase velocity~$c_\mathrm{p}$ regardless of the chosen wavelength discretization~$N_\lambda$.
At the same time, in the source~\cite{Schneider:2010}, the value \mbox{$S_\mathrm{c}=1$} is called ``magic'', which is valid only for vacuum, but not in the general case.
Here we also note that our choice~\eqref{eq:Magic-Courant} is equivalent to setting \mbox{$S_\mathrm{c}=1$} used in the book~\cite{Langtangen:2017}.
\end{remark}

Several examples illustrating the idea of applying the ``magic'' Courant number~\eqref{eq:Magic-Courant} to correct the numerical dispersion in nondispersive homogeneous media, limited by the choice of parameters~\eqref{eq:Definition-Area}, are given in Fig.~\ref{fig:5}.

\begin{figure}
\center
\includegraphics[width=0.75\linewidth]{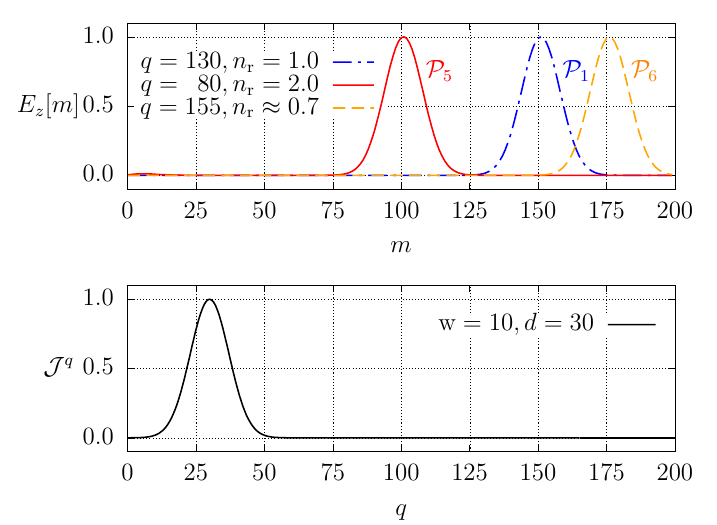}
\caption{Simulating of Gaussian pulse propagation in media optically denser~$\mathcal{P}_5$ (\mbox{$n_\mathrm{r}=2$}) and less dense~$\mathcal{P}_6$ (\mbox{$n_\mathrm{r}=\sqrt{1/2}$}) than vacuum, with using correction of numeric dispersion~(\ref{eq:Magic-Courant}).}
\label{fig:5}
\end{figure}

Fig.~\ref{fig:5} confirms the idea formulated in the form of Statement~\ref{Stmnt:Magic-Courant}~--- numerical dispersion is indeed absent here both in the case of a medium optically denser compared to vacuum (pulse~$\mathcal{P}_5$) and in the case of an optically less dense medium (pulse~$\mathcal{P}_6$).
Fig.~\ref{fig:5} should be compared with Fig.~\ref{fig:1}, in which numerical dispersion is clearly manifested in the shape of the pulse~$\mathcal{P}_2$.

\begin{remark}
Note also that the propagation of the electromagnetic field in space over time (determined, as always, by the value of the discrete index~$q$), presented in Fig.~\ref{fig:5} for the pulse~$\mathcal{P}_5$, occurs with a visible lead compared to the similar process for the pulse~$\mathcal{P}_2$ in Fig.~\ref{fig:1} by a factor of two.
It should be emphasized that this lead is only apparent, since its cause is precisely the twofold difference in the Courant numbers chosen in the case of Fig.~\ref{fig:1} and Fig.~\ref{fig:5}.
In reality (taking into account the correction for~$S_\mathrm{c}$), all the temporal characteristics of the signals in Figs.~\ref{fig:1} and~\ref{fig:5} are identical.
In other words, if we consider the spatial step of the grid~$\Delta_x$ as a fixed parameter that does not change when switching from the simulation in the case \mbox{$S_\mathrm{c}=1$} presented in Fig.~\ref{fig:1} to the simulation with the Courant number~\eqref{eq:Magic-Courant}, then the ``price'' of each time step~$\Delta_t$ (and along with it, the total duration of the simulation) will change by a factor of~$\sqrt{\varepsilon_\mathrm{r}\mu_\mathrm{r}}$.
The same effect (with the replacement of the word ``lead'' by ``lag'') also takes place in the case of optically less dense media (see pulse~$\mathcal{P}_6$) in Fig.~\ref{fig:5}.
\end{remark}

\begin{remark}
We also separately point out that the simulation of the propagation of the pulse~$\mathcal{P}_6$ in a medium optically less dense than vacuum in accordance with the computational algorithm~\eqref{eq:Maxwell-Hy}\,--\,\eqref{eq:Source-Ez} with the choice of the Courant number \mbox{$S_\mathrm{c}=1$} is fundamentally impossible.
It is easy to verify that the calculations in this case very quickly lead to divergences, without giving any useful information, and fundamentally do not describe this particular case.
\end{remark}

The last remark is perfectly illustrated by Fig.~\ref{fig:6}, which shows the avalanche-like process of accumulation of numerical errors during the calculation according to the algorithm~\eqref{eq:Maxwell-Hy}\,--\,\eqref{eq:Source-Ez} in the case of exceeding the Courant number~$S_\mathrm{c}$ compared to~$n_\mathrm{r}$ by only~$0.1\%$.
This illustration is constructed with the same parameters of the signal source~$\mathcal{J}^q$ as in Figs.~\ref{fig:1} and~\ref{fig:5}, for the case of vacuum, although fundamentally the same picture takes place in the case of any other homogeneous nondispersive media from the domain of definition~\eqref{eq:Definition-Area}.

\begin{figure}
\center
\includegraphics[width=0.75\linewidth]{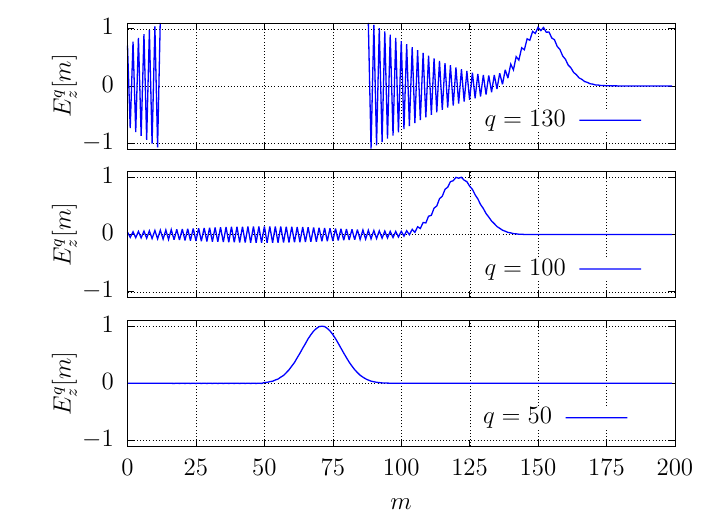}
\caption{Time dynamics of the numerical solution destruction obtained during computation according to the algorithm~(\ref{eq:Maxwell-Hy})\,--\,(\ref{eq:Source-Ez}) when \mbox{$S_\mathrm{c}-n_\mathrm{r}=10^{-3}$}.}
\label{fig:6}
\end{figure}

As can be seen from Fig.~\ref{fig:6}, by the time \mbox{$q=130$} the level of noise that arises randomly as a result of calculations according to the scheme~\eqref{eq:Maxwell-Hy}\,--\,\eqref{eq:Source-Ez} at \mbox{$S_\mathrm{c}-n_\mathrm{r}=10^{-3}$}, twice exceeds the magnitude of the useful signal in absolute value.
At the same time, the spatial extent of the part of the Yee grid affected by this noise is~$3/4$ of the entire length of the computational domain.
The latter circumstance significantly distorts the shape of the trailing edge of the simulated useful signal.
Further, with the subsequent development of the process at \mbox{$q>130$} the useful solution is completely destroyed.

Additional numerical experiments performed by us also show that the effects associated with the accumulation of errors due to self-excitation of the grid always occur when the inequality \mbox{$S_\mathrm{c}>n_\mathrm{r}$} is satisfied.
Thus, at \mbox{$S_\mathrm{c}-n_\mathrm{r}=10^{-4}$} the self-excitation of the grid becomes noticeable already after the passage of the useful signal, but its avalanche-like character is observed in this case as well.
All these observations can be generalized in the form of the following statement.

\begin{proposition}\label{Stmnt:Greater-Courant}
The computational algorithm~\eqref{eq:Maxwell-Hy}\,--\,\eqref{eq:Source-Ez} diverges rapidly and cannot be used to obtain a correct numerical solution of the problem of signal propagation in nondispersive homogeneous media when choosing \mbox{$S_\mathrm{c}>n_\mathrm{r}$}.
\end{proposition}

\begin{proof}
An exhaustive argument for the validity of this statement at the ``physical level of rigor'' is given above when discussing Fig.~\ref{fig:6}.
\end{proof}

It is by applying this statement that the presence of the region~$\mathcal{B}$ shown in Fig.~\ref{fig:4} is explained, although from a formal point of view the function~\eqref{eq:Phi} does not exceed unity in absolute value in this case.

\begin{remark}
At the same time, we note that the choice of values \mbox{$S_\mathrm{c}\leqslant n_\mathrm{r}$}, consistent with the condition \mbox{$\left|\varphi(N_\lambda,S_\mathrm{c},n_\mathrm{r})\right|\leqslant1$}, does not lead to such dramatic consequences as those that take place in Statement~\ref{Stmnt:Greater-Courant}.
\end{remark}

An important consequence of the detailed study carried out in this way, the main results of which are Statements~\ref{Stmnt:Magic-Courant} and~\ref{Stmnt:Greater-Courant}, is one more fact.

\begin{corollary}
Setting the Courant number in the form~\eqref{eq:Magic-Courant} is the only possible optimal choice from the point of view of applying the numerical FDTD algorithm~\eqref{eq:Maxwell-Hy}\,--\,\eqref{eq:Source-Ez} to model the propagation of signals in nondispersive homogeneous media.
\end{corollary}

\begin{proof}
The optimality of such a choice is due to the fact that it eliminates the numerical dispersion when simulating signals of any shape and spectral composition, not contradicting Statement~2 of~\cite{Makarov:2023}, propagating in a wide class of nondispersive homogeneous media described by the domain of definition~\eqref{eq:Definition-Area}.
The uniqueness follows from Statement~\ref{Stmnt:Greater-Courant}.
\end{proof}

\begin{remark}
The problem that has remained unresolved to date (arising from a careful study of the application of the algorithm~\eqref{eq:Maxwell-Hy}\,--\,\eqref{eq:Source-Ez} to homogeneous nondispersive media optically different from vacuum) is that in this case, in addition to the main signal with the selected direction, a signal of relatively small amplitude of the opposite direction is always observed.
Let us call this last signal for brevity~--- backward~$\mathcal{P}_\mathrm{b}$, in contrast to the original~--- forward~$\mathcal{P}_\mathrm{f}$.
An illustration of this problem is given in Fig.~\ref{fig:7}, which is constructed with the same parameters of the signal source as Figs.~\ref{fig:1} and~\ref{fig:5}.
\end{remark}

Fig.~\ref{fig:7} demonstrates the formation of the backward pulse~$\mathcal{P}_\mathrm{b}$ when simulating the propagation of a Gaussian-shaped signal in media optically less dense than vacuum with the optimal choice of the Courant number.
It can be seen that the character of this phenomenon depends significantly on the value of the relative refractive index~$n_\mathrm{r}$ of the medium.
As the latter tends to zero, this phenomenon can no longer be neglected, since it leads to a distortion of the characteristics of the useful forward signal~$\mathcal{P}_\mathrm{f}$ as well.
At the same time, in the case $n_\mathrm{r}>10^{-1}$ the corresponding numerical error is relatively small and can be estimated not to exceed the level of~$1\%$.

\begin{figure}
\center
\includegraphics[width=0.75\linewidth]{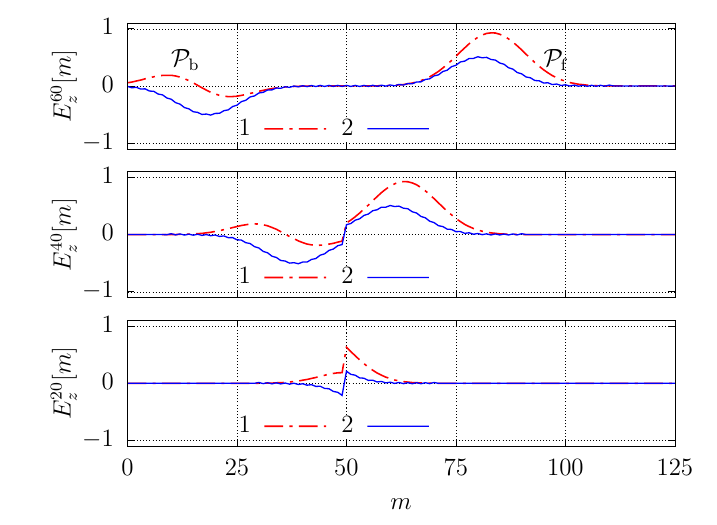}
\caption{Dynamics of forward~$\mathcal{P}_\mathrm{f}$ and backward~$\mathcal{P}_\mathrm{b}$ pulse forming in simulating the propagation of a Gaussian-shaped signal in media optically less dense than vacuum under optimal choice of the Courant number.
The dashed-dotted curve~$1$ corresponds to the case of \mbox{$S_\mathrm{c}=n_\mathrm{r}=10^{-1}$}, the solid line~$2$~--- \mbox{$S_\mathrm{c}=n_\mathrm{r}=10^{-2}$}.}
\label{fig:7}
\end{figure}

\section{Limits of applicability of the main results}

In conclusion of this work, let us discuss the question of the applicability of the results obtained here beyond the domain of definition~\eqref{eq:Definition-Area}.

First, we point out that the interval of values \mbox{$0<\varepsilon_\mathrm{r},\mu_\mathrm{r}<1$} included in the domain of definition~\eqref{eq:Definition-Area}, used when fulfilling the conditions of Statement~\ref{Stmnt:Magic-Courant}, in itself already extends the range of applicability of the computational algorithm developed in this work to the region of exotic media usually not considered in the literature.

Second, we note that our numerical algorithm~\eqref{eq:Maxwell-Hy}\,--\,\eqref{eq:Source-Ez} remains valid also when extending the domain of definition~\eqref{eq:Definition-Area} to the interval \mbox{$\varepsilon_\mathrm{r},\mu_\mathrm{r}\in(-\infty,+\infty)$} under the condition of simultaneous negativity of the medium permittivities \mbox{$\varepsilon_\mathrm{r}\mu_\mathrm{r}>0$}.

Media with simultaneously negative permittivities, as is known~\cite{Veselago,Pendry,Agranovich}, are called left-handed.
In such media, backward waves can propagate~\cite{Shevchenko}, determined by the fact that for them the scalar product of the wave vector~$\mathbf{k}$ (in our work everywhere \mbox{$\mathbf{k}=\beta\mathbf{e}_x$}) and the Umov--Poynting vector~$\mathbf{S}$ is negative
\begin{equation}\label{eq:Backward-Wave}
\left(\mathbf{k} \cdot \mathbf{S}\right) < 0,
\end{equation}
where the energy flux carried by the wave, and determined by the vector~$\mathbf{S}$, is equal to~\cite{Bredov,Kugushev,Shuster,Veselago,Landsberg}
\begin{equation}\label{eq:Umov-Pointing}
\mathbf{S} = \frac{c}{4\pi}\left[\mathbf{E}\times\mathbf{H}\right].
\end{equation}

The validity of such a generalization is demonstrated by Fig.~\ref{fig:8}, which presents the results of comparing the wave packets formed by the same signal source as in Figs.~\ref{fig:1}, \ref{fig:5}\,--\,\ref{fig:7}, recorded at the same time \mbox{$q=100$}, when propagating in a right-handed medium with \mbox{$\varepsilon_\mathrm{r}=\mu_\mathrm{r}=+1$} (bottom half of the figure) and in a left-handed medium with \mbox{$\varepsilon_\mathrm{r}=\mu_\mathrm{r}=-1$} (top half).

\begin{figure}
\center
\includegraphics[width=0.75\linewidth]{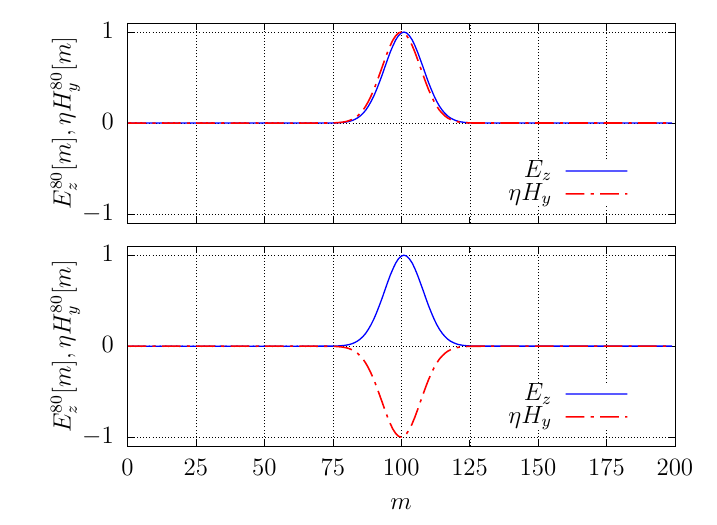}
\caption{Instantaneous snapshots of forward (bottom half) and backward (top half) waves propagating in the Yee grid according to the computational algorithm~(\ref{eq:Maxwell-Hy})\,--\,(\ref{eq:Source-Ez}).}
\label{fig:8}
\end{figure}

It is easy to verify that for the wave packet shown in the upper part of Fig.~\ref{fig:8} the Umov--Poynting vector~\eqref{eq:Umov-Pointing} \mbox{$\mathbf{S}\updownarrows\mathbf{e}_x$}, which is precisely what defines the backward wave according to condition~\eqref{eq:Backward-Wave}.

And finally, let us explain that the validity of the algorithm~\eqref{eq:Maxwell-Hy}\,--\,\eqref{eq:Source-Ez} when extending the domain of definition~\eqref{eq:Definition-Area} to the interval \mbox{$\varepsilon_\mathrm{r},\mu_\mathrm{r}\in(-\infty,+\infty)$} under the condition of non-simultaneous negativity of the medium permittivities \mbox{$\varepsilon_\mathrm{r}\mu_\mathrm{r}<0$} was not investigated in this work.
This is due to the fact that at \mbox{$\varepsilon_\mathrm{r}\mu_\mathrm{r}<0$} the propagation constant~\eqref{eq:Dispersion-Classic} turns out to be an imaginary quantity, which corresponds to a strong attenuation of waves in such media, which as a result turn out to be strongly dispersive (for them, the propagation of plane waves turns out to be impossible, and along with this Statement~\ref{Stmnt:Phase-Velocity} loses its simple meaning, which requires a different formulation in this case).
All this goes beyond the scope of this work.

\section*{Conclusion}

Thus, in this work the phenomenon of numerical dispersion in FDTD modeling of electromagnetic signal propagation in nondispersive homogeneous media has been investigated.
Several statements are formulated that determine the character of this dispersion, and the influence of the Courant number on it is also described.
The optimal value of the Courant number that eliminates the numerical dispersion of wave packets is determined.
The limits of applicability of the developed simulation method are investigated, and for the first time the possibility of its application to media optically less dense than vacuum, as well as to left-handed media, is indicated.

At the same time, the topic of the research still remains quite interesting, since many interesting questions were not investigated in this article, some of which are listed below.

\subsection*{Some open questions}
\begin{enumerate}
\item What can explain from a physical point of view the decrease in the accuracy of the FDTD calculation in media with \mbox{$n_\mathrm{r}<1$}?
An explanation similar to the situation in media for which $n_\mathrm{r}>1$, and based on the decrease in the wavelength in them (and, accordingly, its poor discretization), does not work here.
\item What is the reason for the formation of the backward pulse~$\mathcal{P}_b$ and how can it be eliminated?
Perhaps it is a consequence of some analytical errors made when writing~\eqref{eq:Maxwell-Hy}\,--\,\eqref{eq:Source-Ez} or of the unavoidable errors of floating-point machine representation of numbers in computer calculations.
\item Is it possible to correct numerical dispersion when simulating the propagation of signals in inhomogeneous media?
Can this be achieved by changing the Courant number dynamically during the simulation?
\end{enumerate}

This list of interesting questions does not claim to be complete in any way, and can well be expanded.

\section*{Acknowledgement (state task)}

The work was done in frames of the State task of the Institute of Physics and Mathematics FRC Komi SC UB RAS on the research topic \#~122040400069-8.

\bibliographystyle{unsrtnat}
\bibliography{references}

\end{document}